\documentclass[journal]{IEEEtran}
\ifCLASSINFOpdf
\else
\fi
\usepackage{graphicx}%
\usepackage{multirow}%
\usepackage{amsmath,amssymb,amsfonts}%
\usepackage{amsthm}%
\usepackage{mathrsfs}%
\usepackage{xcolor}%
\usepackage{textcomp}%
\usepackage{manyfoot}%
\usepackage{booktabs}%

\usepackage{listings}%
\usepackage{xcolor}
\usepackage{booktabs}%
\usepackage{array}
\usepackage{rotating}
\usepackage{adjustbox}
\usepackage{pdflscape}  
\usepackage{longtable}  
\usepackage{booktabs} 
\usepackage{geometry} 
\usepackage{url}
\usepackage{subcaption}

\newtheorem{theorem}{Theorem}

\theoremstyle{definition}

\theoremstyle{remark}

\usepackage{algorithm}%
\usepackage{algorithmicx}%
\usepackage{algpseudocode}%

\begin{document}
%
\title{Efficient Shapley values computation for Boolean network models of gene regulation}
%
%
%

\author{Giang~Pham \IEEEauthorrefmark{1}, 
        Silvia~Giulia~Galfrè \IEEEauthorrefmark{1} 
        and~Paolo~Milazzo \IEEEauthorrefmark{1} 
\thanks{This work was supported by the project "Hub multidisciplinare e interregionale di ricerca e sperimentazione clinica per il contrasto alle pandemie e all’antibioticoresistenza (PAN-HUB)” funded by the Italian Ministry of Health (POS 2014-2020, project ID: T4-AN-07, CUP: I53C22001300001) and from a 2023 NARSAD Young Investigator Grant from the Brain \& Behavior Research Foundation.}
\thanks{\IEEEauthorrefmark{1}Department of Computer Science, University of Pisa, Largo B. Pontecorvo 3, 56127 Pisa, Italy.}
}

\maketitle

\begin{abstract}
Identifying dynamically influential nodes in biological networks is a central problem in systems biology, particularly for prioritizing intervention targets in gene regulatory networks. In this paper, we propose a Shapley-value-based framework for assessing the importance of nodes in a Boolean network with respect to a given target node. The framework comprises two complementary measures: the \emph{Knock-out} and the \emph{Knock-in Shapley values}. Moreover, we present a propagation-based method that enables their efficient computation. By exploiting the logical structure of the network, the method avoids exhaustive simulations. The approach is exact for acyclic networks and provides good approximations for cyclic networks. Evaluation on benchmark models from the Cell Collective database shows that the propagation method accurately recovers node importance rankings while achieving substantial speed-ups. 
\end{abstract}

\begin{IEEEkeywords}
Boolean networks, node importance, Knock-out, Shapley value, network propagation.
\end{IEEEkeywords}

%
\IEEEpeerreviewmaketitle

\section{Introduction}
Boolean networks (BNs), introduced by Stuart Kauffman~\cite{kauffman1969metabolic}, have become a widely adopted framework for modeling gene regulatory networks (GRNs). Their success is largely due to their conceptual simplicity and their ability to qualitatively capture regulatory interactions among genes, while still preserving essential dynamic features. As a result, BNs have been successfully applied to a broad range of biological systems and processes~\cite{ALBERT20031,Fang2004yeast,Rodriguez2007tcell}. A central application of BN modeling is the identification of potential intervention targets \cite{Weidner2024gatekeep}, which naturally leads to the problem of quantifying the importance or influence of individual nodes.

Early efforts in this direction focused on characterizing node influence at the level of Boolean functions. Shmulevich and Kauffman introduced the notion of \textit{activity} \cite{shmulevich2004activities}, which measures the influence of a variable on the output of a Boolean function. For a Boolean function $f(\mathbf{x})$ with $\mathbf{x} = (x_1, x_2, \dots, x_k)$, the \textit{activity} of a variable $x_j$ is defined as the average of the partial derivative of $f$ with respect to $x_j$ over all possible configurations $\mathbf{x} \in \{0,1\}^k$. This concept was further analyzed in the context of canalizing functions, and it was shown to have a direct impact on the dynamic behavior of BNs. While foundational, this line of work primarily emphasizes mathematical properties of Boolean functions rather than applications to large-scale biological networks. This limitation was addressed by extending the notion of \textit{activity} to quantify node influence in BN dynamics~\cite{ghanbarnejad2012impact}. The study defined the \emph{dynamical impact} (DI) of a node as the fraction of all possible states for which a perturbation of that node propagates through the network within a given time. This DI can be approximated by the principal eigenvector of the activity matrix, where each entry $\alpha_{ij}$ quantifies the influence of node $i$ on node $j$ (i.e., the \textit{activity} of $i$ on the Boolean function $f_j$). 

Subsequent studies shifted toward information-theoretic approach. Heckel \textit{et al.} introduced the concept of \textit{determinative power} (DP) for an input node \cite{Heckel2013}, defined via mutual information between the state of that input and the states of all non-input nodes in the network.
However, this framework was limited to networks with tree-like topologies, assumed the absence of input dependencies and feedback loops. It also restricts DP to input nodes only. These constraints were later relaxed by Matache \textit{et al.}, who extended the notion of DP to general nodes by relaxing the assumption of independence among the inputs of a Boolean function. This allows the method to work with general BNs \cite{Matache2016}. 
Importantly, DP was shown to recover biologically essential nodes~\cite{pentzien2018identification}. This notion is a theoretically grounded, information-based measure, but it requires either the assumption on the independence among input nodes of Boolean functions as in the work of Heckel \textit{et al.} or the extensive information about the probability of each state as in the relaxed version of Matache \textit{et al.} In addition, it is a generic network measure and is not tailored to specific phenotypic markers. 


Despite their differences, these approaches share a common strategy: identifying influential nodes in Boolean networks while avoiding costly simulations through mathematical abstractions that assume product distributions or equal state probabilities. In contrast, we introduced the \textit{Knock-out Shapley value}~\cite{pham2023preliminary,pham2024gene}, which measures node importance directly from systematic knock-out simulations. This approach is intuitive and biologically relevant, as shown on CD4$^+$ T cell models~\cite{pham2024gene}. In this paper, we propose a general Shapley-value-based framework with two measures: the \emph{Knock-out Shapley value} and the \emph{Knock-in Shapley value}. These measures systematically evaluate gene influence in GRNs across different environmental conditions by quantifying each gene’s marginal contribution to a target node over all input configurations. Averaging these contributions provides a principled measure of gene importance for a specific phenotypic outcome.



Although these notions are biologically meaningful, their standard calculation is computationally expensive since both the number and complexity of simulations grow rapidly with network size~\cite{pham2024gene}. Hence, to overcome this limitations, we also introduce a propagation-based method that enables the efficient computation without performing the knock-out/knock-in simulations. The propagation-based method exploits the logical structure of BNs to enable efficient and scalable estimation, thereby making large-scale screening analyses feasible. It successfully recovers node rankings in most benchmark models, achieving average normalized discounted cumulative gain (NDCG) of 0.779 for knock-in and 0.865 for knock-out Shapley value. Moreover, the method improves computational efficiency, reducing the time complexity from $O(n^2 2^n)$ to $O(n 2^n)$. An implementation of the proposed method, together with scripts to reproduce the experiments, is publicly available at \url{https://github.com/giangpth/ShapleyValueForBooleanNetwork.git}.



The remainder of this paper is organized as follows. Section~\ref{sec:background} reviews the Shapley value and BNs. Section~\ref{sec:metho} introduces the \emph{Knock-out, Knock-in Shapley value} and the propagation-based computation method. Section~\ref{sec:eval} evaluates the method and compares it with existing approaches through two case studies. Finally, Section~\ref{sec:conclusion} concludes the paper and discusses future work.

\section{Background}
\label{sec:background}

\subsection{Shapley Value}
The Shapley value is a concept from cooperative game theory that assigns a fair contribution score to each player based on its marginal effect across all coalitions it might participate in. It has been employed as a tool to analyze a wide range of biological networks~\cite{pham2025comprehensive}.
A cooperative game is defined by a pair
$\langle N, v \rangle$, where $N$ is a set of players and
$v : 2^N \rightarrow \mathbb{R}$ is a payoff function\footnote{The variable $N$ introduced in this section is distinct from those defined in Section \ref{sec:BN}. We follow the convention of using notation that is standard in the literature on Shapley values.}.

For a player $i \in N$, the Shapley value is defined as
\begin{equation}
\footnotesize
\phi_i(v) =
\sum_{S \subseteq N \setminus \{i\}}
\frac{|S|!\, (|N| - |S| - 1)!}{|N|!}
\left[ v(S \cup \{i\}) - v(S) \right]
\end{equation}

Intuitively, $\phi_i(v)$ represents the average marginal contribution of player $i$ over all possible ways of forming a coalition from the remaining players. See \cite{pham2025comprehensive} for more details and examples.


\subsection{Boolean Networks with input nodes}
\label{sec:BN}
A BN is defined as $B = \langle N, F \rangle$, where $N = \{x_1, \dots, x_n\}$ is a set of Boolean variables $x_i$ taking values in $\{0, 1\}$, and $F = \{f_{m+1}, \dots, f_n\}$ with $n > m \geq 0$, is a set of Boolean update functions. BNs can be depicted as directed graphs, with nodes representing Boolean variables and edges representing regulatory interactions. In this paper, the terms variable, node and gene are used interchangeably. We consider BNs with input nodes, denoted $B_{n,m}$ ($m > 0$), where the first $m$ variables $x_1, \dots, x_m$ are \emph{input nodes} without update functions, and the remaining $n-m$ variables are \emph{internal nodes} updated with the corresponding function: 
\begin{equation}
\begin{split}
x_i(t+1) &= f_i\big(x_{j_1(i)}(t), \dots, x_{j_{k_i}(i)}(t)\big), \\
\quad i &= m+1, \dots, n .
\end{split}
\end{equation}

Input variables remain fixed during simulation, while internal variables evolve based on the set of Boolean functions. Different update schemes can be employed, including synchronous and asynchronous scheme~\cite{barbuti2020survey,schwab2020concepts}. In the synchronous scheme, all Boolean functions are applied simultaneously at each time step to update the variables.  In the asynchronous scheme, a single Boolean variable is non-deterministically selected and updated at each time step according to its corresponding function. In this study, we apply the synchronous updating scheme due to its deterministic property \cite{barbuti2020survey,schwab2020concepts}. 

The simulation of a BN produces a state graph, where nodes represent the reachable states of the BN, and edges denote state transitions governed by the update functions. \textit{Attractors} correspond to loops within the state graph, capturing the network's long-term behavior, and are often associated with phenotypes. An attractor can be represented as a set of states within the loop, \( A = \{ a_{j+1},a_{j+2}, \dots, a_{j+k} \} \), where the special case \( k = 1 \) corresponds to a singleton attractor, consisting of a single state $a_{j+1}$. 

An \emph{input configuration} $S \subseteq \{x_1, \dots, x_m\}$ is defined
by setting variables in $S$ to $1$ and all remaining input variables to $0$. For each input configuration, the network is initialized with all
internal variables set to $0$ and simulated synchronously until an attractor is reached. For a given input configuration $S$, we denote by $A_{B_{n,m}}(S) = \{a^S_{j+1}, a^S_{j+2}, \dots a^S_{j+k}\}$ the corresponding attractor. For periodic attractors, we summarize the attractor by the maximum value attained by each node: $\mathbf{a}^S = \bigvee_{a^S_i \in A_{B_{n,m}}(S)} a^S_i$. This representation indicates whether a gene is expressed (at some level) within the stable state of the system.

\section{Methodology}
\label{sec:metho}
\subsection{Shapley Values in Boolean Networks}
\label{sec:kogame}
In Boolean networks, Shapley values can be used to quantify the importance of nodes in activating a designated \emph{target node}. The payoff function is defined on input configurations and determined by the value of the target node in the resulting attractor.

\paragraph*{The Knock-out game}
A preliminary version of the \emph{Knock-out game}  was introduced by our earlier works \cite{pham2023preliminary, pham2024gene}. The method evaluates how disabling individual nodes affects the activation of the target under different input configurations.

A Knock-out game is defined as a pair \( \langle \mathbb{B}_{n,m}, \texttt{T} \rangle \), where \( \mathbb{B}_{n,m} \) is a BN and \( \texttt{T} \) is an internal target node. The \textit{payoff function} \( v \) is determined by the value of \( \texttt{T} \) in the attractor reached from each input configuration \( S \), and is defined as
\begin{equation}
\label{eq:payoff}
v^{\texttt{T}}_{\mathbb{B}_{n,m}}(S)
= \mathbf{a}^S(\texttt{T}) = \bigvee_{a^S_i \in  A_{\mathbb{B}_{n,m}}(S)} a^S_i(\texttt{T})
\end{equation}

\paragraph*{The Knock-out Shapley value}
For a node \( \texttt{X} \) in $\mathbb{B}_{n,m}$, let \( \mathbb{B}_{n,m}^{\bar{\texttt{X}}} \) denote the modified network in which \( \texttt{X} \) is permanently fixed to \(0\). The Knock-out Shapley value of \( \texttt{X} \) is defined as
\begin{equation}
\label{eq:shapko}
\phi^{ko}_{\texttt{X}}(v)
= \sum_{S \subseteq M} w_{S}
\bigl[
v^{\texttt{T}}_{\mathbb{B}_{n,m}}(S)
-
v^{\texttt{T}}_{\mathbb{B}_{n,m}^{\bar{\texttt{X}}}}(S)
\bigr],
\end{equation}
where \( M = \{x_1, x_2, \dots, x_m\} \) denotes the set of input nodes and $w_S$ denotes the weight associated with input configuration $S$. Following the original Shapley value formulation, we use
$w_S = \frac{|S|! (|M| - |S|)!}{|M|!}$.
Alternative weighting schemes, such as uniform weights $w_S = k$, or task-specific designs that emphasize coalitions containing particularly relevant inputs, can also be adopted. This flexibility allows the method to be tailored to different goals.


Similar to the Knock-out game, the \textit{Knock-in game} is defined by a pair \( \langle \mathbb{B}_{n,m}, \texttt{T} \rangle \), where \( \mathbb{B}_{n,m} \) is a Boolean network and \( \texttt{T} \) is an internal target node. The payoff function is defined similarly to that of the Knock-out game and is determined by the value of \( \texttt{T} \) in the attractor reached from each input configuration.

For a node \( \texttt{X} \in \mathbb{B}_{n,m} \), let \( \mathbb{B}_{n,m}^{\mathring{\texttt{X}}} \) denote the modified network in which \( \texttt{X} \) is permanently fixed to \(1\). The knock-in Shapley value of \( \texttt{X} \) is defined as
\begin{equation}
\label{eq:shapki}
\phi^{ki}_{\texttt{X}}(v)
= \sum_{S \subseteq M} w_{S}
\bigl[
v^{\texttt{T}}_{\mathbb{B}_{n,m}}(S)
-
v^{\texttt{T}}_{\mathbb{B}_{n,m}^{\mathring{\texttt{X}}}}(S)
\bigr]
\end{equation}


\subsection{Propagation method for the Knock-out and Knock-in Shapley values computation}
\label{sec:rules}
The propagation method we are going to define assumes to be applied to BNs with a special structure, called Binarized Boolean Networks (BBNs). In a BBN, each node has at most two incoming edges (i.e., each Boolean function is defined on at most two Boolean variables) in which, negation regulation has exactly one input node. In this section, we present the theoretical framework for propagation in such networks. The procedure for transforming a general BN into a BBN is described in Section~\ref{sec:implement}.

\subsubsection{Simulation result as a truth table}
The simulation of a BN can be represented as a truth table, where the number of rows corresponds to the number of input configurations. Each row of the truth table is assigned a weight corresponding to the associated term in the summations of Equations~\ref{eq:shapko} and~\ref{eq:shapki}. 
For instance, for a Boolean network with three input nodes, as illustrated in Figure~\ref{fig:simnet}, there are $2^3$ input configurations, corresponding to 8 rows in Table~\ref{tab:simnet}. Each row represents the singleton attractor or the overall state $\mathbf{a}_S$ of the periodic attractor obtained from the corresponding input configuration as defined in Section \ref{sec:background}.
This representation enables direct and convenient computation of the \textit{payoff function} defined in Equation~\ref{eq:payoff}.

From Equations \ref{eq:shapko} and \ref{eq:shapki}, the Knock-out and Knock-in Shapley values can be interpreted as the sum of the weighted differences between the value of the target node $\texttt{T}$ under two scenarios: the intact network $\mathbb{B}_{n,m}$ and the network $\mathbb{B}_{n,m}^{\bar{\texttt{X}}}$ where $\texttt{X}$ is fixed to 0 (Knock-out) or $\mathbb{B}_{n,m}^{\mathring{\texttt{X}}}$ where $\texttt{X}$ is fixed to 1 (Knock-in), evaluated across all input configurations. This difference is nonzero only for those input configurations $S$ in which fixing the value of $\texttt{X}$ alters the value of the target node $\texttt{T}$ relative to the intact network. For all other configurations, the difference is zero and does not contribute to the Shapley value. 
With this formulation, the computation of the Knock-out and Knock-in Shapley values reduces to identifying the rows for which
$
v_{\mathbb{B}_{n,m}}^{T}(S) - v_{\mathbb{B}_{n,m}^{\bar{X}}}^{T}(S) \neq 0
$
in the Knock-out, and
$
v_{\mathbb{B}_{n,m}}^{T}(S) - v_{\mathbb{B}_{n,m}^{\mathring{X}}}^{T}(S) \neq 0
$
in the Knock-in.

\begin{figure}[ht]
\centering
\begin{minipage}{0.30\textwidth}
    \centering
    \includegraphics[width=0.6\linewidth]{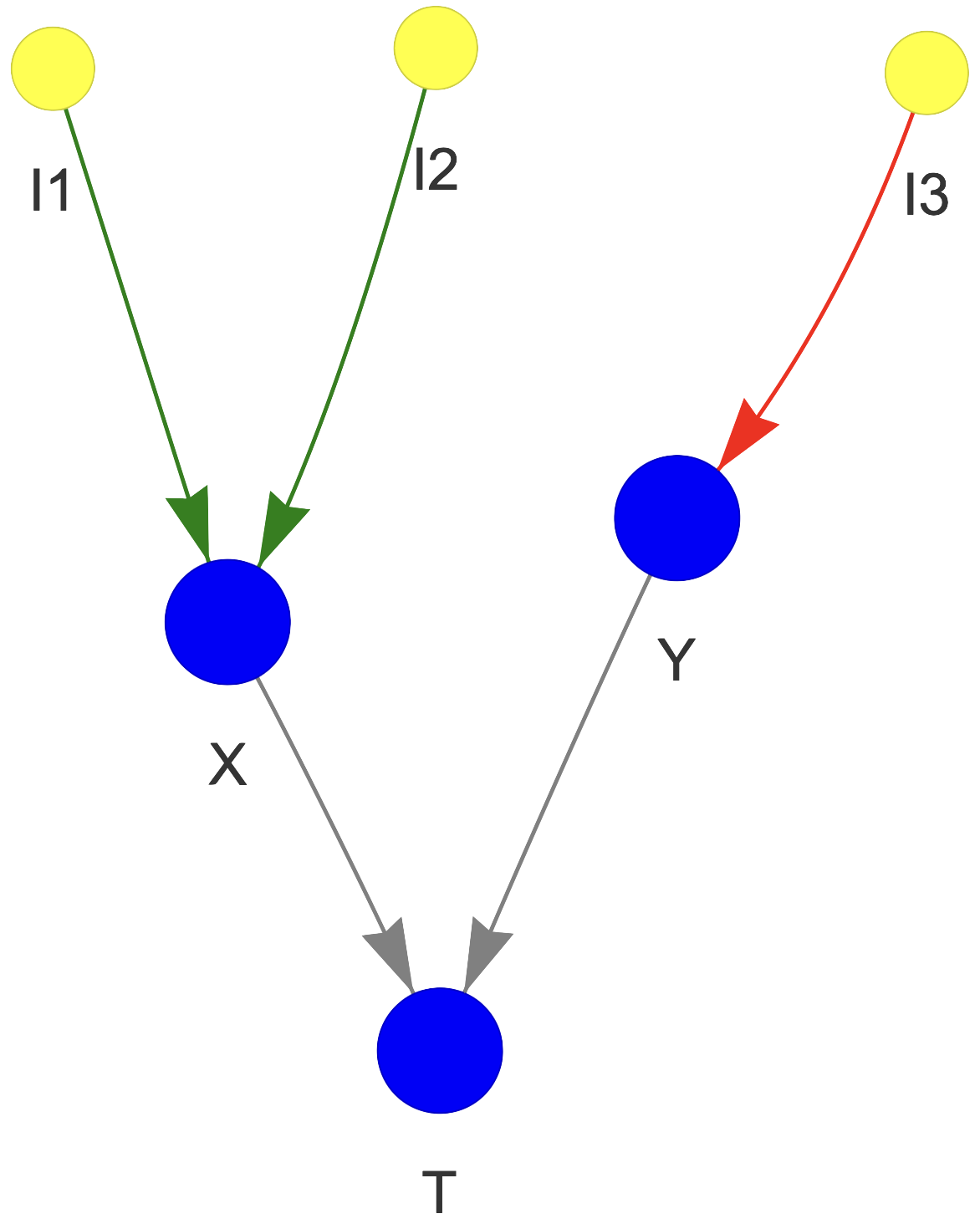}
    \end{minipage}
    \begin{minipage}{0.17\textwidth}
    \small
    \texttt{X = I1 AND I2}\\
    \texttt{Y = NOT I3  }\\
    \texttt{T = X OR Y}
    \end{minipage}
    \caption{A Binary Boolean Network with the corresponding Boolean functions} 
     \label{fig:simnet}
\end{figure}

\begin{table}[]
    \centering
    \begin{tabular}{|c||ccc|ccc|}
        \hline
        Weight &I1 & I2 & I3 & X & Y & T \\
        \hline
         $w_1= 1$ & 0 & 0 & 0 & 0 & 1 & 1 \\
         $w_2 = 1/3$ & 0 & 0 & 1 & 0 & 0 & 0 \\
         $w_3 = 1/3$ & 0 & 1 & 0 & 0 & 1 & 1 \\
         $w_4 = 1/3$ & 0 & 1 & 1 & 0 & 0 & 0 \\
         $w_5 = 1/3 $ & 1 & 0 & 0 & 0 & 1 & 1 \\
         $w_6 = 1/3 $ & 1 & 0 & 1 & 0 & 0 & 0 \\
         $w_7 = 1/3$ & 1 & 1 & 0 & 1 & 1 & 1 \\
         $w_8 = 1$ & 1 & 1 & 1 & 1 & 0 & 1 \\
        \hline
    \end{tabular}
    \caption{Truth table of the network in Figure \ref{fig:simnet}.}
    \label{tab:simnet}
\end{table}

For rows in which the value of \(\texttt{X}\) is equal to $1$ in the truth table, fixing \(\texttt{X}\) to $1$ leaves the network unchanged, and therefore
$
v_{\mathbb{B}_{n,m}}^{\texttt{T}}(S) - v_{\mathbb{B}_{n,m}^{\mathring{\texttt{X}}}}^{\texttt{T}}(S) = 0 .
$
As a result, these rows do not contribute to the Knock-in Shapley value, but they may contribute to the Knock-out Shapley value when \(\texttt{X}\) is fixed to $0$. Conversely, for rows in which \(\texttt{X}\) is equal to $0$, contributions arise for the Knock-in value but not for the Knock-out value. 
The Shapley values can then be computed by summing the signed weights associated with the corresponding sets of rows. For a node \texttt{X}, if its value in row $w_i$ is $1$, the corresponding weight contributes to the Knock-out value; if it is $0$, it contributes to the Knock-in value. The sign of each weight is determined by the change in the value of the target node: if the target is originally $1$ and the modification of \texttt{X} causes it to switch to $0$, the weight is assigned a positive sign, and conversely, a negative sign is assigned when the modification induces a transition from $0$ to $1$.


\subsubsection{Propagation for Boolean operators}
Since computing the Knock-in and Knock-out Shapley values reduces to identifying rows where flipping a node changes the target node, the procedure starts at the target node (\texttt{T}), where all rows are relevant as any change in its value trivially affects its own state. These rows are then propagated upward through the BBN to the input nodes. Because each node in a BBN depends on at most two parent nodes, the propagation rules can be defined simply and systematically.
\paragraph{\texttt{OR} operator: $ \texttt{A} = \texttt{B} \texttt{ OR } \texttt{C}$}
For binary operator \texttt{OR} of the form $ \texttt{A} = \texttt{B} \texttt{ OR } \texttt{C}$, changing the value of \( \texttt{B}\) affects the value of \(\texttt{A}\) only when \(\texttt{C} = 0\), since \(\texttt{A} = 1\) whenever \(\texttt{C} = 1\), regardless of the value of \(\texttt{B}\). Consequently, among the rows associated with node \(\texttt{A}\), only those rows for which \(\texttt{C} = 0\) are propagated to node \(\texttt{B}\), and analogously, rows with \(\texttt{B} = 0\) are propagated to node \(\texttt{C}\). For example, for the \texttt{OR} operator
\(\texttt{T = X  OR  Y}\) in Figure \ref{fig:simnet}, with \(\texttt{T}\) as the target node, all rows are initially associated with \(\texttt{T}\), corresponding to the weight set \(\{ w_1, w_2, w_3, w_4, w_5, w_6, w_7, w_8 \}\). The subset of rows propagated to node \(\texttt{X}\) from \(\texttt{T}\) is \(\{ w_2, w_4, w_6, w_8 \}\), corresponding to configurations in which \(\texttt{Y} = 0\). Analogously, the subset of rows propagated to node \(\texttt{Y}\) is \(\{ w_1, w_2, w_3, w_4, w_5, w_6 \}\), corresponding to configurations in which \(\texttt{X} = 0\). 

\paragraph{\texttt{AND} operator: $\texttt{A = B  AND  C}$} 
For a binary  operator \texttt{AND} of the form $\texttt{A = B  AND  C}$, the rows propagated from \(\texttt{A}\) to node \(\texttt{B}\) are those for which \(\texttt{C} = 1\), since only in this case, changes in \(\texttt{B}\) can influence the value of \(\texttt{A}\). An analogous rule applies to the rows propagated from \(\texttt{A}\) to node \(\texttt{C}\). Consider the network shown in Figure~\ref{fig:simnet} and the operator \(\texttt{X = I1  AND  I2}\). If the set of rows associated with node \(\texttt{X}\) is \(\{ w_2, w_4, w_6, w_8 \}\), then the subset of rows propagated to node \(\texttt{I1}\) is \(\{ w_4, w_8 \}\), corresponding to
configurations in which \(\texttt{I2} = 1\), while the subset of rows propagated to node \(\texttt{I2}\) is \(\{ w_6, w_8 \}\), in which \(\texttt{I1} = 1\).

\paragraph{Identity operator: $\texttt{A = B}$; negation operator: $\texttt{A =  NOT B}$ } 
In the case of an identity operator, the set of rows associated with \(\texttt{B}\) is identical to that associated with \(\texttt{A}\). This follows directly from the fact that any change in \(\texttt{B}\) induces a corresponding
change in \(\texttt{A}\), and thus \(\texttt{A}\) propagates the effect of \(\texttt{B}\) to the target node \(\texttt{T}\). Similarly, for the negation operator \(\texttt{A = \texttt{NOT} B}\), node \(\texttt{B}\) also inherits the full set of rows associated with \(\texttt{A}\). The difference lies in the roles of these rows: rows contributing to the knock-out value of \(\texttt{A}\) contribute to the knock-in value of \(\texttt{B}\), and vice versa. In the network of Figure~\ref{fig:simnet}, consider the relation \(\texttt{Y = \texttt{NOT} I3}\). In this case, node \(\texttt{I3}\) inherits the full set of rows associated with node \(\texttt{Y}\), namely \(\{ w_1, w_2, w_3, w_4, w_5, w_6 \}\). However, rows contributing to the Knock-out value of \(\texttt{Y}\), given by \(\{ w_1, w_3, w_5 \}\) (for which \(\texttt{Y} = 1\)), contribute instead to the Knock-in value of \(\texttt{I3}\) (for which \(\texttt{I3} = 0\)). The remaining rows, \(\{ w_2, w_4, w_6 \}\), correspondingly contribute to the Knock-out value of \(\texttt{I3}\).


For the network in Figure~\ref{fig:simnet}, the results are: for \texttt{I1}, $\mathrm{KO}=1.0$ and $\mathrm{KI}=-0.3333$; for \texttt{I2}, $\mathrm{KO}=1.0$ and $\mathrm{KI}=-0.3333$; for \texttt{I3}, $\mathrm{KO}=-1.0$ and $\mathrm{KI}=1.6667$; for \texttt{T}, $\mathrm{KO}=3.0$ and $\mathrm{KI}=-1.0$; for \texttt{X}, $\mathrm{KO}=1.0$ and $\mathrm{KI}=-1.0$; and for \texttt{Y}, $\mathrm{KO}=1.6667$ and $\mathrm{KI}=-1.0$. These results are consistent with those obtained through direct simulation.

\subsubsection{Non-trivial structures}
A BBN may contain non-trivial structures that are not captured by the proposed rules. In particular, they are diamond-shaped convergences and cycles, in which cycles introduce approximation in the propagation results.

\paragraph{Diamond structures}
Beyond simple unary and binary operator structures, a BBN may also contain diamond-shaped motifs, as illustrated in Figure~\ref{fig:diamond} (with $\texttt{I}$ is the target node), in which a node \(\texttt{B}\) influences a downstream node \(\texttt{G}\) through two distinct incoming paths via nodes \(\texttt{D}\) and \(\texttt{E}\). In this case, the set of rows propagated to \(\texttt{B} \) cannot be obtained simply by joining the set of rows propagated from \(\texttt{D}\) and \(\texttt{E}\).
Indeed, the set of rows associated with node \(\texttt{B}\) must first be a subset of the rows associated with node \(\texttt{G}\), since \(\texttt{B}\) influences the target node \(\texttt{I}\) through \(\texttt{G}\). However, because the value of \(\texttt{G}\) depends on both \(\texttt{B}\) (via nodes \(\texttt{D}\) and \(\texttt{E}\)) and \(\texttt{A}\) (via node \(\texttt{D}\)), not all rows associated with \(\texttt{G}\) contribute to the influence of \(\texttt{B}\). To address this case, a straightforward approach is to partially simulate the diamond structure, which in this example involves the four nodes  \(\texttt{B}\), \(\texttt{D}\), \(\texttt{E}\), and \(\texttt{G}\). To identify the rows associated with \(\texttt{B}\), we simulate flipping the value of \(\texttt{B}\) in each row associated with \(\texttt{G}\) and retain only those rows for which this perturbation leads to a change in \(\texttt{G}\). 

In a general BBN, diamond structures can be non-trivial because nodes involved in cycles may induce self-influencing diamonds, as illustrated in Figure~\ref{fig:diamondcycle}. In this example, in addition to the diamond formed by \texttt{A}, \texttt{B}, \texttt{C}, and \texttt{D}, there is another diamond originating from \texttt{D}, passing through \texttt{A}, \texttt{B} in one side, and \texttt{C} in the other side, and feeding back into \texttt{D}. Such configurations significantly complicate the treatment of diamonds. Thus, we restrict the diamonds to acyclic BBNs. More precisely, we identify diamond structures on an acyclic approximation of the original network obtained by removing all cycles. The details of this procedure are discussed in Section~\ref{sec:implement}. Removing cycles prior to diamond identification may cause some diamond structures to be missed, constituting the first source of approximation in our propagation method.

\begin{figure}
    \centering
    \begin{subfigure}[t]{0.235\textwidth}
        \centering
         \includegraphics[width=0.9\linewidth]{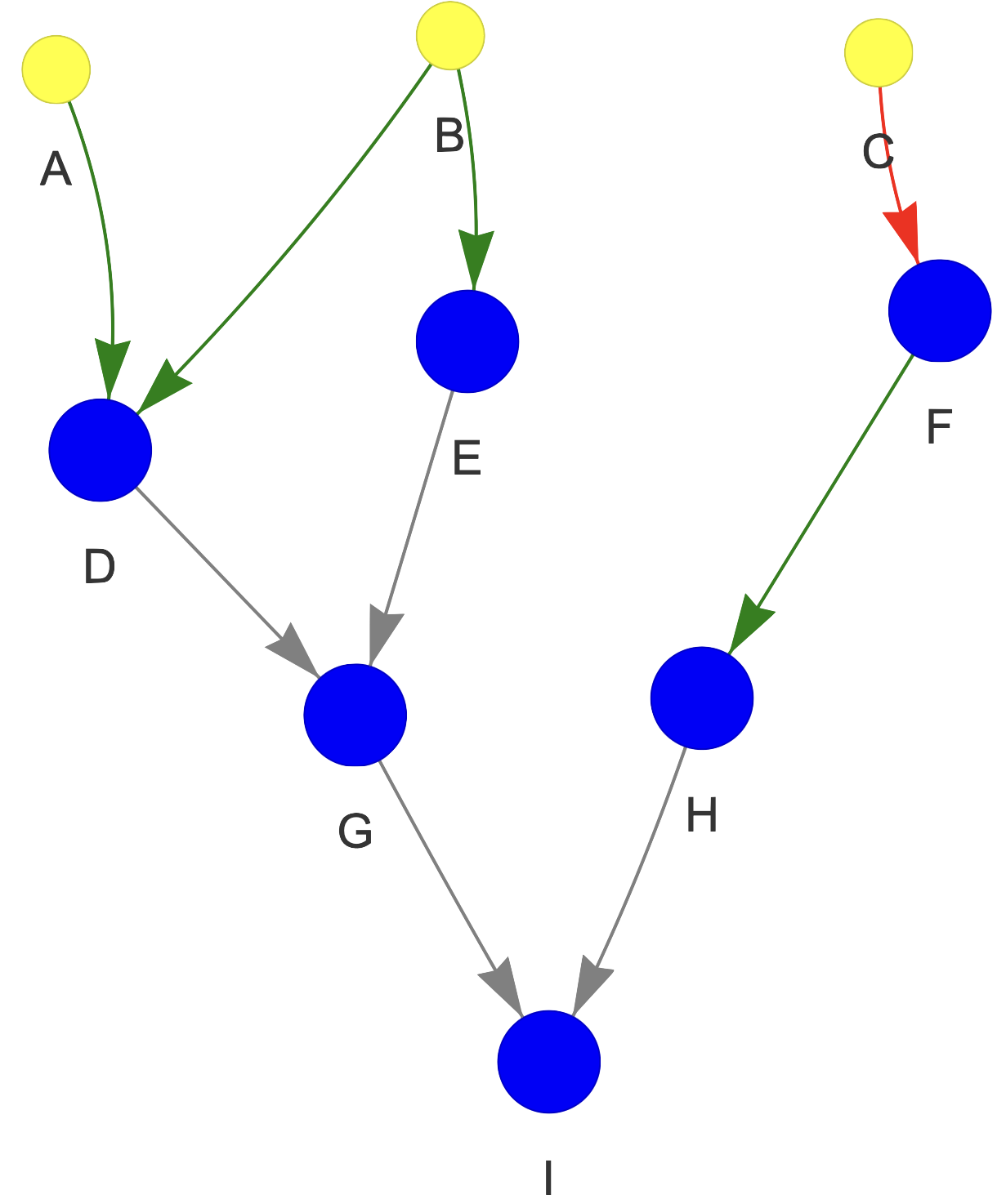}
    \caption{An acyclic BBN with diamond structure including \texttt{B, D, E} and \texttt{G}}
    \label{fig:diamond}
    \end{subfigure}
    \hfill
    \begin{subfigure}[t]{0.235\textwidth}
        \centering
        \includegraphics[width=0.76\linewidth]{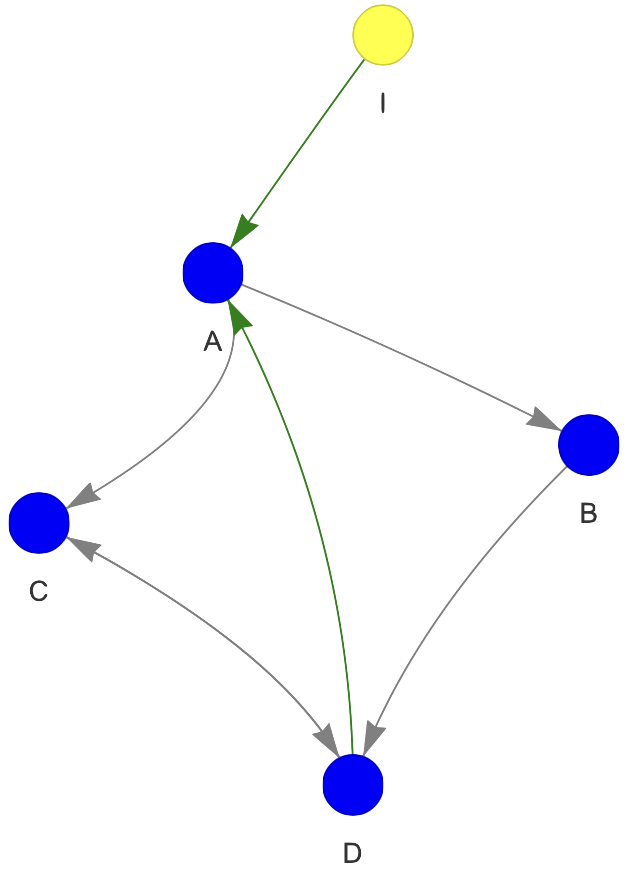}
    \caption{A cyclic BBN with diamond structures.}
    \label{fig:diamondcycle}
    \end{subfigure}
\end{figure}

\paragraph{Cycles}
The rules defined for binary operators assume independence among the operands and the result node. However, this assumption may be violated in networks with cycles. For instance, in the BBN shown in Figure~\ref{fig:cycle}, the cycle involving \texttt{G}, \texttt{I}, and \texttt{E} induces a dependency of \texttt{G} on \texttt{I} through \texttt{E}. This dependency can lead to incorrect row propagation. When \texttt{I} is selected as the target node, it is initially associated with all eight rows of the truth table in Table~\ref{tab:cycle}. According to the propagation rules, only rows with \(\texttt{G} = 0\) should be propagated to node \texttt{H}. However, because \texttt{G} and \texttt{I} are part of a cycle together with \texttt{E}, some rows with \(\texttt{G} = 1\) are nevertheless associated with \texttt{H}. In particular, \(w_1\) and \(w_5\) should be accounted for at \texttt{H} despite \(\texttt{G} = 1\), since in these configurations the value of \texttt{G} is influenced by \texttt{I}, which in turn is influenced by \texttt{H}.

To address this, one approach is to explicitly simulate nodes involved in cycles, as done for diamond structures; however, in complex networks this may require simulating the entire network, reducing computational gains. Hence, when encountering cycles, we consider two strategies: (i) marking nodes belonging to cycles as visited after the first propagation step and stop the propagation there, or (ii) allowing propagation to continue until convergence (or until a visitation threshold is reached). For the second strategy, propagated rows within cycles can either be merged or intersected, resulting in three approaches in total. All approaches were evaluated experimentally as shown in Figures~S4 and S5, Supplementary Material (SM). The convergence-based approach with merged the sets of rows (Method~2) performs slightly better and is therefore adopted.

\begin{figure}[ht]
\centering
\begin{minipage}{0.32\textwidth}
   \centering
    \includegraphics[width=0.6\linewidth]{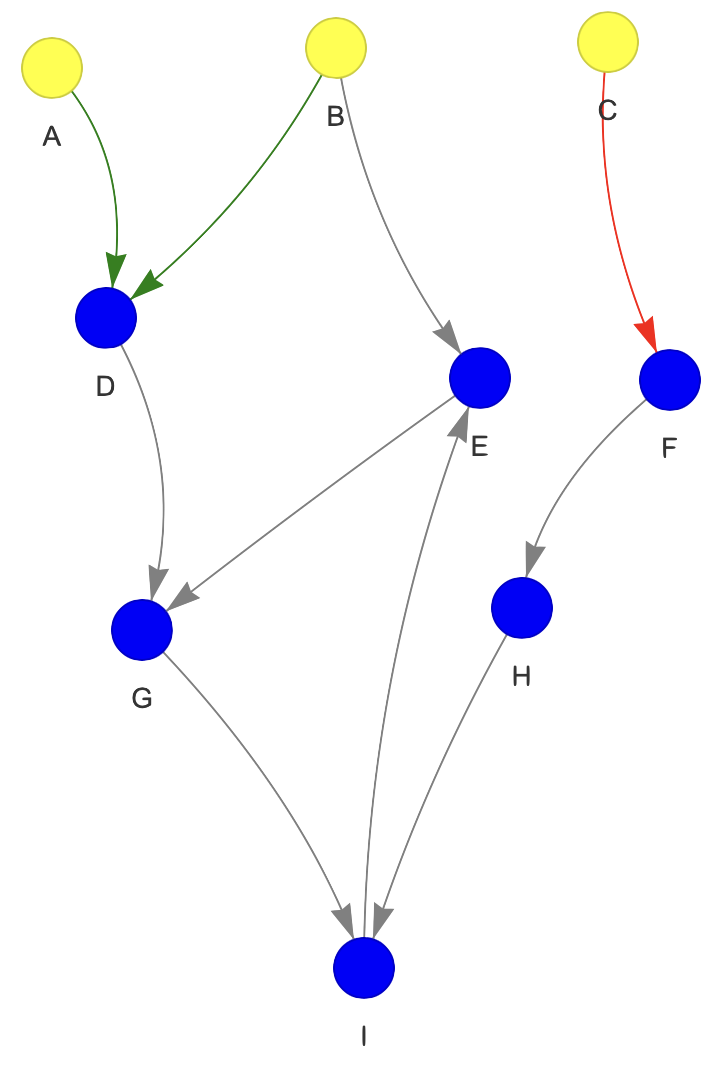}
    \end{minipage}
    \begin{minipage}{0.15\textwidth}
    \small
        \texttt{D = A AND B} \\
         \texttt{E = B OR I} \\
         \texttt{F = NOT C}\\
          \texttt{G = D OR E} \\
           \texttt{H = F} \\ 
            \texttt{I = G OR H }
    \end{minipage}
 \caption{A BBN network with cycle and corresponding Boolean functions}
  \label{fig:cycle}
\end{figure}
 
\begin{table}[]
    \centering
    \begin{tabular}{|c||ccc|ccccc|c|}
        \hline
        Weight & A & B & C & D & E & F & G & H & I \\
        \hline
         $w_1 $ & 0 & 0 & 0 & 0 & 1 & 1 & 1 & 1 & 1 \\
         $w_2 $ & 0 & 0 & 1 & 0 & 0 & 0 & 0 & 0 & 0 \\
         $w_3 $ & 0 & 1 & 0 & 0 & 1 & 1 & 1 & 1 & 1 \\
         $w_4 $ & 0 & 1 & 1 & 0 & 1 & 0 & 1 & 0 & 1 \\
         $w_5 $ & 1 & 0 & 0 & 0 & 1 & 1 & 1 & 1 & 1 \\
         $w_6 $ & 1 & 0 & 1 & 0 & 0 & 0 & 0 & 0 & 0 \\
         $w_7 $ & 1 & 1 & 0 & 1 & 1 & 1 & 1 & 1 & 1 \\
         $w_8 $ & 1 & 1 & 1 & 1 & 1 & 0 & 1 & 0 & 1 \\
        \hline
    \end{tabular}
    \caption{Truth table of the network in Figure \ref{fig:cycle}.}
    \label{tab:cycle} 
\end{table}

\paragraph*{Complexity analysis}
The propagation method avoids repeatedly simulating the perturbed network for each internal node, reducing the complexity from ($O(n^2 2^n)$) to ($O(n2^n)$). Details of the complexity analysis are provided in Section S5 of the SM.

\subsection{Correctness in acyclic networks}
\label{sec:proof}
Since all approximations arise from cycles in the BBN, the propagation method is exact for acyclic networks. This is validated empirically on synthetic acyclic networks (see SM, Figure~S9). We now formally establish the correctness of the propagation method in the acyclic case. The complete proof is provided in the SM; here, we present a sketch of the argument.

\begin{theorem}
\label{thm:acyclic_correctness}
Let $\mathbb{B}_{n,m}$ be an acyclic binarized Boolean network and let $\texttt{T}$ be a target node.
For any node $\texttt{X}$, the knock-out and knock-in Shapley values computed by the proposed propagation method are equal to those obtained by explicit simulation. 
\end{theorem}

\begin{proof}
Restrict to the backward reachable subgraph $\mathbb{B}^T$ from target node \texttt{T}, since nodes outside it have zero Shapley value. Because the graph is acyclic, $\mathbb{B}^\texttt{T}$ admits a topological order toward $\texttt{T}$, and $\texttt{T}$ can be written as a Boolean formula obtained by composing node functions along this order. Each propagated row corresponds to a coalition.

At $\texttt{T}$, propagation enumerates the truth table and counts exactly the coalitions that flip the output under knock-in/knock-out, matching the definitions. Assume correctness for all nodes closer to $\texttt{T}$ than $\texttt{Y}$. For node $\texttt{Y}$, there are two cases:
(i) If $\texttt{Y}$ starts a diamond, partial simulation ensures each coalition affecting $\texttt{T}$ is counted for $\texttt{Y}$, so contributions match explicit simulation.
(ii) Otherwise, all paths through $\texttt{Y}$ go through a unique successor $\texttt{X}$; changing $\texttt{Y}$ affects $\texttt{T}$ iff it affects $\texttt{X}$. Since $\texttt{X}$ is correct and the propagation computes Boolean updates correctly (Lemma 5, SM), $\texttt{Y}$ is also correct. By induction, the propagation matches explicit simulation on all nodes in $\mathbb{B}^\texttt{T}$, completing the proof.
\end{proof}


\subsection{Implementation details}
\label{sec:implement}
The propagation rules are defined for BBNs. We therefore first transform a general Boolean network into an equivalent BBN in which (i) each node has at most two inputs, and (ii) any node with a negatively influential input has exactly one incoming edge. This transformation preserves the dynamic behavior of the original network.

To identify diamond structures, the network must be acyclic. We therefore remove a feedback arc set (FAS) from the original network, as it is smaller than the binarized one, and then binarize the resulting acyclic network to obtain an approximate acyclic BBN. Diamond structures are detected on this acyclic BBN. We detect the FAS in the original network using Algorithm~S1 (SM) of Eades, Lin, and Smyth~\cite{eades1993fast}.



\paragraph{``Binarization''}
We binarize the network by introducing intermediate nodes so that each node has at most two incoming edges while preserving functionality. For a node \(\texttt{A}\) with operands \(\texttt{I}_1, \texttt{I}_2, \texttt{I}_3\), an intermediate node \(\texttt{I}_{23}\) aggregates \(\texttt{I}_2\) and \(\texttt{I}_3\) is introduced. \(\texttt{A}\) then receives inputs from \(\texttt{I}_1\) and \(\texttt{I}_{23}\). Boolean functions are updated accordingly, e.g., \(\texttt{I}_{23} = \texttt{I}_2 \texttt{ OR } \texttt{I}_3\) and \(\texttt{A} = \texttt{I}_1 \texttt{ OR } \texttt{I}_{23}\). To preserve synchronous dynamics, all intermediate nodes are updated before the original nodes.

\paragraph{Identification of diamond structures}
We detect diamonds by traversing the BBN backward from a designated output node. Each node propagates a carry-on map indicating whether upstream nodes are reached via the left or right branch of binary operators. When multiple paths converge, conflicts in the carry-on states indicate potential diamond closures; nodes with multiple outgoing edges are recorded as sinks. Unary operators propagate the carry-on information unchanged. The full procedure is given in SM Algorithm S2. Among the detected structures, we select the deepest (most downstream) node as the ``master diamond'' and discard nested smaller diamonds, since simulating the master diamonds is sufficient to determine the corresponding set of rows for the initial node.

\paragraph{Propagation in the BBNs}
The resulted BBN is still acyclic as the FAS had been removed before the binarization. We reintroduce the removed arcs back to the BBN before performing the propagation. During the propagation, the rules in Section~\ref{sec:rules} are applied as presented in Algorithm~S3 (see SM). Besides the BBN, set of master diamonds, target node, and Boolean formulas, the algorithm takes two indices derived from the truth table: an index \(I\), mapping each node to rows where its value is $1$, and a complementary index \(\bar I\), mapping to rows where its value is $0$. These inverted indices enable efficient row queries.

Starting from the target node \({\texttt{T}}\), the algorithm propagates layer by layer, computing for each node \(\texttt{X}\) the set of truth-table rows \(R(\texttt{X})\) determining its Shapley values. 

\subsection{Further improvements}
Random simulations can improve accuracy in cyclic networks by explicitly simulating a subset of nodes during propagation instead of applying the rules in Section~\ref{sec:rules}. Experiments (five runs per model-target pair; Figures S1–S2) show that both runtime and the fraction of correctly propagated rows increase roughly linearly with the proportion of simulated nodes, with runtime growing faster. This reduces efficiency but is beneficial when higher accuracy is required.

Figure S3 in the SM shows that the fraction of correctly propagated rows at input nodes correlates with that across all nodes. Since exact rows at input nodes can be obtained from intact network simulations, they provide a reference to estimate propagation error.

\section{Evaluation of the propagation method}
\label{sec:eval}
We evaluated the proposed propagation method using a benchmark of \textcolor{black}{20} models (18 cyclic and 2 acyclic) obtained from the Cell Collective database\footnote{\url{https://cellcollective.org}}. Models were randomly selected under the constraint that they contain between two and fifteen input nodes. Models exceeding this range were excluded due to their higher computational demands for the baseline simulations, as the experiments involve numerous repeated runs. For each model, three target nodes were chosen according to the model descriptions (e.g., their biological functions) if available; otherwise randomly. This resulted in 60 data points of (model, target) pairs. 

\subsection{Accuracy}
We assess propagation performance using relative RMSE (RMSE normalized by the maximum value) and the normalized discounted cumulative gain (NDCG) over the full ranking length as primary metrics, both ranging from 0 to 1 with 1 is optimal value. The average results of the benchmark is show in Table \ref{tab:average}. Kendall’s \(\tau_b\) and Spearman’s \(\rho\) are additionally reported as auxiliary ranking-based measures (Figure S6, SM). Except for relative RMSE, these metrics evaluate ranking quality rather than exact values, as relative gene importance is more informative than absolute scores. Finally, the fraction of correctly propagated rows (from 0 to 1) is reported in Figure~S7 (SM) to directly measure the propagation accuracy.

\begin{table}[]
    \centering
    \begin{tabular}{cc|cc|c}
        \multicolumn{2}{c|}{RMSE} & \multicolumn{2}{c|}{NDCG} & Ratio \\
        KO & KI & KO & KI &  \\
        \hline
        0.0195 & 0.0288 & 0.865 & 0.779 & 11.28
    \end{tabular}
    \caption{Average values of the metrics for KO and KI, and runtime ratio between the simulation-based analysis and the propagation method.}
    \label{tab:average}
\end{table}

\begin{figure*}[t]
    \centering
    \includegraphics[width=\textwidth]{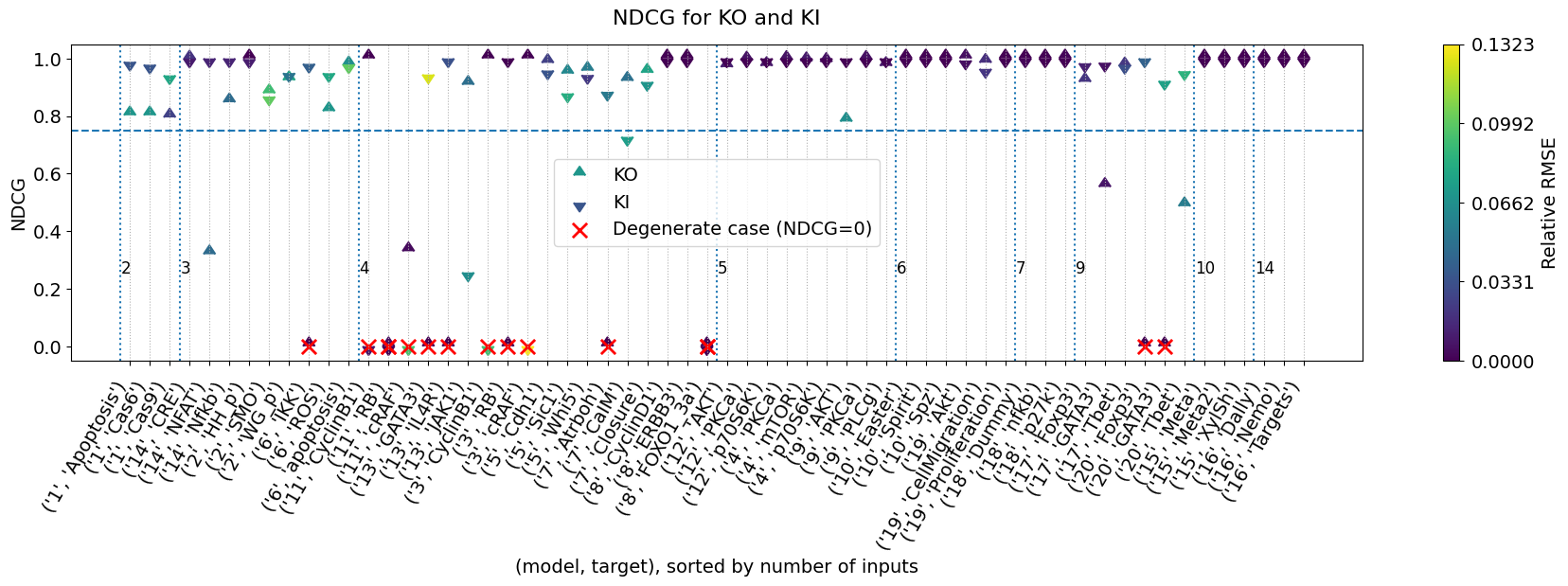}
    \caption{NDCG of KO and KI values and relative RMSE under the Shapley weighting scheme. Models are grouped by input size (within each dashed vertical segment, labeled by the number of inputs). The horizontal line indicates the NDCG threshold of 0.75, while degenerate cases, where all nodes share the same ranking, are marked with red crosses.}
    \label{fig:ndcg}
\end{figure*}

Figure~\ref{fig:ndcg} presents the NDCG results for the Knock-out (KO) and Knock-in (KI) Shapley values under the Shapley weighting scheme; uniform weighting (Figure~S8, SM) shows similar trends. The average relative RMSE is 0.0194 for KO and 0.0288 for KI. The average NDCG is 0.779 for KO and 0.865 for KI (Table \ref{tab:average}).

Using an NDCG threshold of 0.75 (blue dashed line), 13 KO and 8 KI samples fall below this value. Most of these correspond to degenerate baseline cases in which all nodes have identical scores, resulting in identical ranks and an NDCG of zero (cases below the red-crossed line). Excluding degenerate cases, most model-target pairs achieve NDCG well above 0.75, often close to 1, indicating good ranking recovery. KI is more stable, with tightly concentrated NDCG values and low relative RMSE, whereas KO shows larger variability and more frequent degeneracy. Performance exhibits a weak dependence on the number of network inputs.

\subsection{Time performance}
\begin{figure*}[t]
    \centering
    \includegraphics[width=\textwidth]{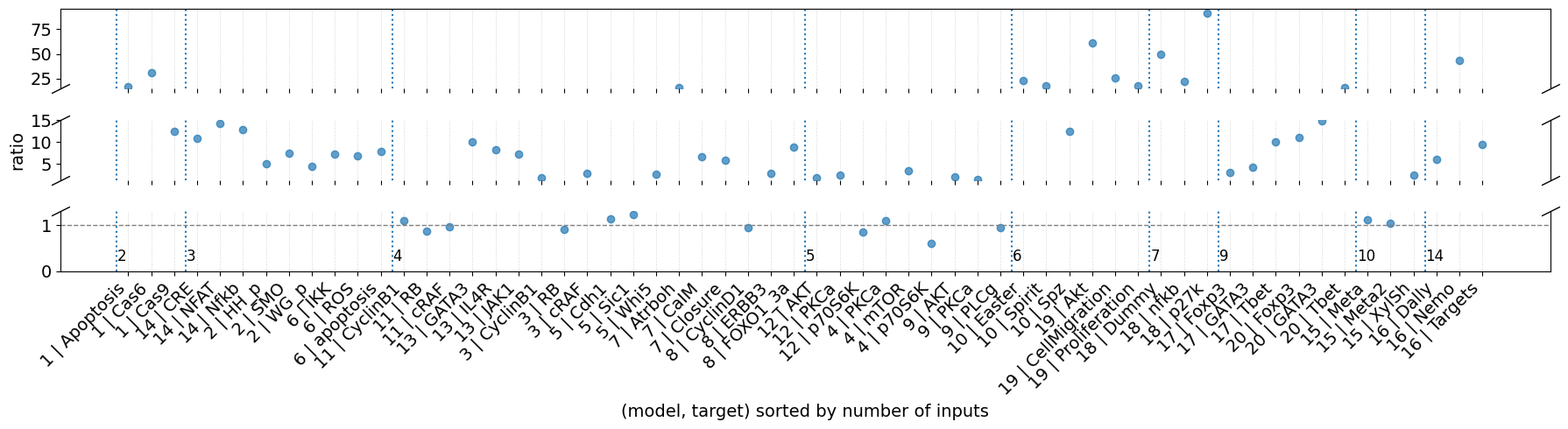}
    \caption{Ratio between the runtime of the standard analysis and the propagation method. Models are grouped by input size (within each dashed vertical line, labeled by the number of inputs).}
    \label{fig:time_shap}
\end{figure*}

Figure~\ref{fig:time_shap} shows the ratio between the runtime of the standard analysis and the propagation-based method with the average of 11.28 (Table \ref{tab:average}), with values greater than one indicating a speed-up. The plot is split into three panels to highlight the low-ratio range (0–1.3) and the middle range (1.3-15) where many points fall into; the horizontal line at 1 marks equal runtime. 

Overall, propagation is faster in most cases, with speed-ups increasing with the number of inputs. Small models show only modest gains, while larger and more complex networks often achieve speed-ups of one order, and in some cases nearly two orders of magnitude, reflecting the good scalability of the propagation approach. Only a few small models exhibit ratios close to or below one. Specifically, 7 (model, target) pairs have ratios below 1, all in models with four or five inputs; the minimum ratio ($\approx 0.6$) occurs for target \texttt{p70S6K} in model~4, which completes under one second. In larger models, only targets \texttt{Meta} and \texttt{Meta2} in model~15 show limited gains, due to relatively smaller number of internal nodes in comparison with the total number of nodes in the network. Even in this case, propagation remains faster, with ratios of 1.115 and 1.04, respectively.

\subsection{Comparison with other notions} 
Here we compare our results with existing methods such as DP \cite{Heckel2013},  \textit{Strength} \cite{Matache2016} and DI \cite{ghanbarnejad2012impact} when available. It is important to note that such notions produce a single generic gene ranking per model, whereas our method yields target-specific ranking.

\paragraph{Fibroblast signaling pathway \cite{helikar2008emergent,puniya2016systems}}
Helikar \textit{et al.} proposed a large-scale Boolean model of signal transduction in a generic cell type with 139 nodes, in which 9 nodes are input nodes\footnote{\url{https://research.cellcollective.org/dashboard#module/1557:1/}} \cite{helikar2008emergent}. The model integrates a large number of highly connected pathways including receptor tyrosine kinase
G protein-coupled receptors and the integrin signaling pathways. This model was extensively analyzed by Puniya \textit{et al.} through 1300 simulations to identify the most influential components based on their impact on the rest of the system \cite{puniya2016systems}. 

When DP and \textit{Strength} are applied to this pathway, DP ranks \texttt{EGFR}, \texttt{ASK1}, \texttt{Src}, \texttt{PIP3\_345}, and \texttt{PKC} as the top nodes \cite{Matache2016,pentzien2018identification}, while \emph{Strength} identifies \texttt{Src}, \texttt{PIP3\_345}, \texttt{PKC}, \texttt{PIP2\_45}, and \texttt{EGFR} \cite{Matache2016}. The DI ranks \texttt{Src}, \texttt{B-Arrestin}, \texttt{GRK}, \texttt{PIP2\_45}, and \texttt{PKC} among the top five nodes (in that order) \cite{ghanbarnejad2012impact}. These results are all consistent with the findings of Puniya~\cite{puniya2016systems}. The approximation of the DI based on the principal eigenvector of the activity matrix also identifies four over the top five nodes, with the exception of \texttt{GRK} \cite{ghanbarnejad2012impact}. 

Unlike these generic ranking methods, the proposed approach captures target-specific influence, enabling a more precise identification of regulatory roles. Considering cellular growth, \texttt{Akt} and \texttt{Erk} (MAPK1) are used as targets due to their established role as growth markers~\cite{puniya2016systems}. The analysis identifies \texttt{PIP3\_345} (Phosphatidylinositol (3,4,5)-triphosphate) and \texttt{PIP2\_45} (Phosphatidylinositol-4,5-biphosphate) as key regulators, together with nodes such as \texttt{ILK}, \texttt{PI5K} (Phosphatidylinositol 4-phosphate 5-kinase), \texttt{PA} (Phosphatidic acid), \texttt{Rho}, \texttt{RhoK} (GRK1), \texttt{ARF} (CDKN2A), and \texttt{PLD} (PLD1) for \texttt{Akt}, and \texttt{IL1\_TNF}, \texttt{Mek} (MAP2K1), \texttt{PP2A}, \texttt{Tpl2}, \texttt{Trafs} (heterogeneous groups stimulated by IL1\_TNFR), \texttt{IL1\_TNFR}, and \texttt{Raf} (ZHX2) for \texttt{Erk} (MAPK1). For \texttt{Akt}, the prominence of \texttt{PIP3\_345} and \texttt{PIP2\_45} is biologically coherent, since \texttt{PI3K} converts \texttt{PIP2} into \texttt{PIP3} and \texttt{PIP3} serves as the membrane-localized signal that recruits and activates \texttt{Akt}; in other words, the phosphoinositide layer is expected to dominate any target-specific ranking centered on \texttt{Akt} output. 
The additional appearance of \texttt{ILK, PI5K, PLD, PA, Rho, RhoK}, and \texttt{ARF} suggests that the method is capturing a broader integrin/membrane-lipid/cytoskeletal control layer around \texttt{Akt}, not just the canonical receptor-to-PI3K step; this interpretation is consistent with the known role of \texttt{ILK} in anchorage-dependent growth and Akt-related signaling \cite{gorska2022integrin}, and with the fact that \texttt{PLD} generates phosphatidic acid while Rho-family, ARF-family and GTPases regulate PLD-associated lipid signaling~\cite{bruntz2014phospholipase}.
Thus, several of these components are consistent with previously reported key regulators of growth dynamics~\cite{puniya2016systems}. Moreover, this target-oriented analysis provides a more detailed and nuanced profile of node importance compared to the generic score.

\paragraph{T-cell Receptor signaling network  \cite{Rodriguez2007tcell}} 
This Boolean network describes T cell activation mechanism via the T cell receptor, the CD4/CD8 co-receptors, and the accessory signaling receptor CD28 with 94 nodes and 123 interactions\footnote{\url{https://research.cellcollective.org/dashboard#module/2171:1/}}. For this model, DP ranks \texttt{pkb} (AKT1), \texttt{lat}, \texttt{fyn}, \texttt{zap70}, and \texttt{lckp1} (CD4-dependent pool of lck) as the top five influential nodes \cite{pentzien2018identification}. These are membrane-proximal control nodes and are therefore expected to dominate generic rankings. The target-specific Shapley analysis answers a narrower question: which perturbations most affect a single downstream output, here \texttt{jnk} (c-Jun N-terminal kinases, principally MAPK8 and MAP9). Taking \texttt{jnk} (principally MAPK8 and MAP9) as the target (which was extensively analyzed in \cite{Rodriguez2007tcell}), \texttt{cd28} shows a weak influence on \texttt{jnk} (KO: $0.4002$; KI: $-0.5335$). In this context, the weak but non-zero influence of \texttt{cd28} is biologically coherent. The original logical model predicted that \texttt{cd28} stimulation alone selectively activates \texttt{jnk}; importantly, this prediction was then validated experimentally by the observation of evident and sustained \texttt{jnk} phosphorylation after \texttt{cd28} stimulation~\cite{weiss2000regulation}. By contrast, \texttt{pi3k} exhibits no effect (KO = 0; KI = 0). This is consistent with the model prediction and \textit{in vitro} experiments reported in \cite{Rodriguez2007tcell}, which indicates that \texttt{cd28} stimulation induces activation of \texttt{jnk}, while \texttt{jnk} activation does not depend on \texttt{pi3k} (PIK3CD) within the \texttt{cd28}-driven signaling context. This should be interpreted in a branch-specific rather than universal sense. Rodriguez et al. showed that under \texttt{cd28} stimulation, \texttt{pi3k} inhibition blocked \texttt{pkb} phosphorylation but not \texttt{jnk} phosphorylation, which is consistent with the present target-specific result. However, the same study also reported that under \texttt{tcr} stimulation, \texttt{jnk} activation was blocked by \texttt{pi3k} inhibition, and therefore concluded that the model likely lacked a direct or indirect between \texttt{pi3k} and \texttt{jnk} in that context.
The Shapley analysis does not suggest that \texttt{pi3k} is irrelevant to T-cell signaling in general; rather, it suggests that \texttt{pi3k} is dispensable for the \texttt{cd28} to \texttt{jnk} branch captured here, even though \texttt{pi3k} remains central to other T-cell outputs and may still contribute to \texttt{jnk} regulation in other stimulation contexts.
Overall, this case study illustrates the main conceptual advantage of the proposed framework: global influence measures emphasize the general control architecture of the T-cell receptor network, whereas the target-specific Shapley analysis resolves which branch is most relevant for a chosen phenotype. Here, it shows that although proximal hubs dominate generic rankings, for the specific output \texttt{jnk}, \texttt{cd28} exerts limited but genuine control, whereas  \texttt{pi3k} is not required in the CD28-centered branch represented by this model output~\cite{shah2021t}.


\section {Conclusion and future work}
\label{sec:conclusion}
This paper proposes a framework to compute the Knock-out and Knock-in Shapley values in BNs as a measurement for node importance. By exploiting the logical structure of binarized networks, the method avoids exhaustive simulations while preserving target-specific importance semantics. Experiments on benchmark models show that the approach recovers node rankings with good accuracy and achieves substantial speed-ups, particularly in networks with many inputs. Despite limitations in cyclic networks, the method provides a practical and scalable alternative to simulation-based Shapley analysis for BN models.


Several directions can be taken to strengthen the proposed framework. First, the handling of diamond structures can be further optimized. Instead of partially simulating entire diamond subgraphs, future work will focus on detecting the precise entry and exit points of signal interaction within diamonds and on characterizing their logical influence patterns analytically. By examining the conditions under which signals from multiple paths actually interact at the reconvergence node, it should be possible to reduce or eliminate local simulations, thereby improving runtime. Secondly, we will address cyclic networks by deriving explicit error bounds for propagation in the presence of feedback. By quantifying how cycles distort independence assumptions and by relating propagation errors to structural properties of the cycles (e.g., length, connectivity, or feedback strength), we aim to provide guarantees on the maximum deviation from exact Shapley values. Such bounds would clarify the trade-off between efficiency and accuracy and guide the use of partial simulation strategies in cyclic networks.





\ifCLASSOPTIONcaptionsoff
  \newpage
\fi



\bibliographystyle{ieeetr}
\bibliography{bibtex/bib/IEEEexample}
%


%


\begin{IEEEbiographynophoto}{Giang Pham}
Giang Pham received her M.Sc. degree in Computer Science from University of Pisa, Italy. She is currently pursuing a Ph.D. degree in Computer Science at the University of Pisa, Italy. Her work focuses on investigating the theoretical foundations and developing practical frameworks for analyzing biological networks.
\end{IEEEbiographynophoto}


\begin{IEEEbiographynophoto}{Silvia Giulia Galfrè}
Silvia Giulia Galfrè is a bioinformatician and researcher specializing in computational biology and single-cell RNA sequencing data analysis. She earned her PhD in Cellular and Molecular Biology with a focus on bioinformatics from the University of Rome Tor Vergata. Her work centers on developing computational methods and tools, including the COTAN package, to investigate gene expression and regulatory mechanisms. Currently, she is involved in biomedical data analysis projects and interdisciplinary collaborations, combining expertise in biology, statistics, and machine learning.
\end{IEEEbiographynophoto}

\begin{IEEEbiographynophoto}{Paolo Milazzo}
Paolo Milazzo is Associate Professor at the Department of Computer Science of the University of Pisa. His research primarily focuses on computational modeling and network analysis techniques for the study of complex systems, with emphasis on Systems Biology applications. He has authored more than 130 peer-reviewed publications in international journals and conference proceedings and serves as a program committee member for various conferences in the fields of theoretical computer science and computational biology.
\end{IEEEbiographynophoto}


\vfill



\end{document}